\documentclass[11pt,english]{article}
\usepackage[latin9]{inputenc}
\usepackage{amsmath}
\usepackage{amsthm}
\usepackage{amssymb}

\makeatletter
\theoremstyle{plain}
\newtheorem{thm}{\protect\theoremname}[section]
  \theoremstyle{plain}
  \newtheorem{cor}[thm]{\protect\corollaryname}
  \theoremstyle{remark}
  \newtheorem{rem}[thm]{\protect\remarkname}
  \theoremstyle{plain}
  \newtheorem{conjecture}[thm]{\protect\conjecturename}
  \theoremstyle{definition}
  \newtheorem{problem}[thm]{\protect\problemname}
  \theoremstyle{remark}
  \newtheorem{claim}[thm]{\protect\claimname}

\usepackage{complexity,color,bm}
\usepackage{colortbl}

\providecommand{\E}{\mathrm{E}}

\definecolor{gray-comment}{gray}{0.5}

\theoremstyle{plain}

\newclass{\UPP}{UPP}

\DeclareMathOperator*{\ED}{ED}
\DeclareMathOperator*{\LCS}{LCS}
\newclass{\AMA}{AMA}

\newtheorem{oq}{Open Question}

\makeatother

\usepackage{babel}
  \providecommand{\claimname}{Claim}
  \providecommand{\conjecturename}{Conjecture}
  \providecommand{\corollaryname}{Corollary}
  \providecommand{\problemname}{Problem}
  \providecommand{\remarkname}{Remark}
\providecommand{\theoremname}{Theorem}

\begin{document}

\title{Hardness of Approximate Nearest Neighbor Search}

\author{Aviad Rubinstein\thanks{Harvard University \texttt{aviad@seas.harvard.edu}. This research was supported by a Rabin Postdoctoral Fellowship.
I thank Amir Abboud, Karl Bringmann, and Pasin Manurangsi for encouraging me to write up this paper. I am also grateful to Amir, Lijie Chen, and Ryan Williams for inspiring discussions. Thanks also to Amir, Lijie, Vasileios Nakos, Zhao Song and anonymous reviewers for comments on earlier versions. Finally, this work would not have been possible without the help of Gil Cohen and Madhu Sudan with AG codes.}}
\maketitle
\begin{abstract}
We prove conditional near-quadratic running time lower bounds for
approximate Bichromatic Closest Pair with Euclidean, Manhattan, Hamming,
or edit distance. Specifically, unless the Strong Exponential Time
Hypothesis (SETH) is false, for every $\delta>0$ there exists a constant
$\varepsilon>0$ such that computing a $\left(1+\varepsilon\right)$-approximation
to the Bichromatic Closest Pair requires $\Omega\left(n^{2-\delta}\right)$
time. In particular, this implies a near-linear query time for Approximate
Nearest Neighbor search with polynomial preprocessing time.

Our reduction uses the Distributed PCP framework of \cite{ARW17-proceedings},
but obtains improved efficiency using Algebraic Geometry (AG) codes.
Efficient PCPs from AG codes have been constructed in other settings
before \cite{BKKMS16-linear_PCP,BCGRS16-IOP}, but our construction
is the first to yield new hardness results.

\thispagestyle{empty}

\newpage{}
\end{abstract}

\section{Introduction}

Approximate Nearest Neighbor (ANN) search is an important problem
in practice as well as a fundamental problem in theory: preprocess
a set $A$ of $N$ vectors so that given a query vector $b$, an (approximately)
closest vector ($a^{*}\in A$ such that $\left\Vert a^{*}-b\right\Vert \approx\min_{a\in A}\left\Vert a-b\right\Vert $)
can be found efficiently. In this paper we prove conditional hardness
results for the easier offline variant, called Batch Approximate Nearest
Neighbor or {\sc Bichromatic Closest Pair}\footnote{{\em Bichromatic} distinguishes this problem from the yet easier {\em Monochromatic} Closest Pair problem, where the input is just one set of vectors and the algorithm can return any pair of vectors.}:
given sets $A,B$ of $N$ vectors, find a pair $a^{*}\in A$ and $b^{*}\in B$
that is approximately the closest (i.e. $\left\Vert a^{*}-b^{*}\right\Vert \approx\min_{\substack{a\in A\\
b\in B
}
}\left\Vert a-b\right\Vert $). Since this problem is easier, our hardness results extend to Approximate
Nearest Neighbor search as well.

Many algorithms have been designed for Approximate Nearest Neighbor
search and Closest Pair. Even for simple $\ell_{p}$ metrics, there
are also many impossibility results, including: bounds on specific
algorithmic approaches \cite{MNP07-LSH_lower_bound,OWZ14,ALRW16},
unconditional lower bounds in the cell-probe and related models (but
those hold only against a very small number of cell probes) \cite{AIP06,PT07-probe,ACP08,PTW08,CR10-probe,PTW10,KP12,AV15,ALRW16,LPY16-probe},
and conditional hardness assuming SETH\footnote{The Strong Exponential Time Hypothesis (SETH) postulates that for
every $\varepsilon$ there is a $k=k\left(\varepsilon\right)$ such
that $k$-SAT over $n$ variables requires $\left(2-\varepsilon\right)^{n}$
time.} of the exact variants \cite{Wil05,AW15,DSL16-bichromatic,Wil18}. 

We are particularly interested in subquadratic algorithms that give
$\left(1+\varepsilon\right)$-approximation to {\sc Bichromatic Closest Pair},
for arbitrary small $\varepsilon>0$. For simplicity, we limit the
following discussion to Euclidean metric, but similar results hold
for Hamming and Manhattan metrics as well. LSH based techniques can
solve this problem in $O\left(N^{2-\Theta\left(\varepsilon\right)}\right)$
time \cite{IM98}, but not faster \cite{MNP07-LSH_lower_bound,OWZ14}.
Valiant's algorithm obtained an improved runtime of $O\left(N^{2-\Theta\left(\sqrt{\varepsilon}\right)}\right)$
\cite{Valiant15}. The state of the art is an $O\left(N^{2-\tilde{\Theta}\left(\varepsilon^{1/3}\right)}\right)$-time
algorithm by Alman, Chan, and Williams \cite{ACW16-ANN-epsilon3}.
Can the dependence on $\varepsilon$ be improved indefinitely? Note,
for example, that the main result in \cite{Valiant15} is an algorithm
for an average case variant of the same problem which, for any constant
$\varepsilon$, runs in time $O\left(N^{\frac{5-\omega}{4-\omega}}\right)=O\left(N^{1.62}\right)$\footnote{The running time was later improved to $O\left(N^{\frac{2\omega}{3}}\right)=O\left(N^{1.59}\right)$
by \cite{KarKK16}.} (where $\omega$ is the exponent of matrix multiplication). Can we
hope for an $f\left(\varepsilon\right)\cdot N^{1.99}$ time algorithm
(for any function $f$) for the general case? Our main result in this
paper rules out such algorithms:
\begin{thm}
\label{thm:intro-main}Assuming SETH, for every constant $\delta>0$
there exists a constant $\varepsilon=\varepsilon\left(\delta\right)$
such that approximating {\sc Bichromatic Closest Pair} in Euclidean,
Manhattan, or Hamming distance to within $\left(1+\varepsilon\right)$
requires $O\left(N^{2-\delta}\right)$ time.
\end{thm}

We also derive an analogous hardness result for {\sc Bichromatic Closest Pair}
with edit distance\footnote{The {\em edit distance}, sometimes also Levenshtein distance, between
two strings $a,b$ is defined as the minimum number of character deletions,
insertions, or substitutions required to transform $a$ to $b$.}. However, we note that the best algorithms for approximate {\sc Bichromatic Closest Pair}
with edit distance \cite{ostrovsky2007low} are far from matching
our hardness result.
\begin{thm}
\label{thm:intro-edit}Assuming SETH, for every constant $\delta>0$
there exists a constant $\varepsilon=\varepsilon\left(\delta\right)$
such that approximating {\sc Bichromatic Closest Pair} in edit distance
to within $\left(1+\varepsilon\right)$ requires $O\left(N^{2-\delta}\right)$
time.
\end{thm}

As we mentioned earlier, these results also imply the first (to our
knowledge) conditional hardness of Approximate Nearest Neighbor search.
In particular, they imply a near-linear lower bound on query time
for any data structure with a polynomial preprocessing time:
\begin{cor}
\label{cor:ANN}Assuming SETH, for every constants $\delta,c>0$ there
exists a constant $\varepsilon=\varepsilon\left(\delta,c\right)$
such that no algorithm can preprocess a set of $N$ vectors in $O\left(N^{c}\right)$
time, and subsequently answer $\left(1+\varepsilon\right)$-Approximate
Nearest Neighbor queries in $O\left(N^{1-\delta}\right)$ time (for
Euclidean, Manhattan, Hamming, or edit distance).
\end{cor}

\begin{rem}
Informally, in terms of the dependence between $\delta$ and $\varepsilon$,
our reduction loses roughly a single exponential factor ($\delta\approx\frac{\log\log1/\varepsilon}{\log1/\varepsilon}$;
see Equation (\ref{eq:delta-epsilon})); this is already quite far
from best known algorithm \cite{ACW16-ANN-epsilon3}, which obtains
$\delta=\tilde{\Theta}\left(\varepsilon^{1/3}\right)$. Formally,
things are worse: we are unable to deduce any specific relation between
$\delta$ and $\varepsilon$ from SETH, because our reduction also
depends on the unspecified relation between $\epsilon$ and $k$ in
the formulation of SETH (alternatively, $\delta$ and $c$ in the
formulation of the Orthogonal Vectors Conjecture (OVC); see Conjectures
\ref{conj:SETH} and \ref{conj:OVC}).
\end{rem}

\subsection{Techniques}

The main technical idea in this paper is a simple instantiation the
Distributed PCP framework from \cite{ARW17-proceedings} with Algebraic
Geometry (AG) codes.

\subsubsection*{Distributed PCP}

Most reductions from SETH to polynomial time problems use the following
outline: Given a CNF $\varphi$ over $n$ variables, construct gadget(s)
for each half assignment $\alpha,\beta\in\left\{ 0,1\right\} ^{n/2}$.
If the number of gadgets is $N\approx2^{n/2}$, then SETH implies
an $\approx N^{2}$ lower bound on the running time for solving the
new instance. One obstacle for using Probabilistically Checkable Proofs
(PCP) to prove hardness of approximation in \P~is the {\em PCP blowup}:
given a $3$-CNF $\varphi$ over $n$ variables, the PCP Theorem \cite{AS98-PCP,ALMSS98-PCP}
gives an efficient construction of a $3$-SAT $\varphi'$ over $n'$
variables which is hard to approximate. But if we construct $N'\approx2^{n'/2}\gg2^{n}$
gadgets, we do not obtain any meaningful running time lower bounds
(even assuming SETH). 

Our starting point is the {\em Distributed PCP} framework of \cite{ARW17-proceedings}.
Another way to think about the PCP Theorem is as a way to, given assignment
$x\in\left\{ 0,1\right\} ^{n}$ to CNF $\varphi$, generate a proof
$\pi\left(x\right)$ that $x$ satisfies $\varphi$. The proof $\pi\left(x\right)$
should be ``probabilistically checkable'', which means that a randomized
verifier needs to read only a small number of bits from the proof.
\cite{ARW17-proceedings} observed that if we construct the proof
$\pi\left(x\right)$ in a distributed manner, namely construct a ``half-proof''
$\pi\left(\alpha\right),\pi\left(\beta\right)$ for each ``half-assignment'',
we overcome the PCP blowup because the number of gadgets remains small.
Constructing distributed PCP in this way is not quite possible, but
\cite{ARW17-proceedings} gave a construction using a short non-deterministic
hint. Since we can enumerate over all short hints in deterministic
subexponential time, this construction gives near-quadratic hardness
of approximation for several problems such as {\sc Maximum Inner
Product} over $\left\{ 0,1\right\} $-vectors.

The techniques of \cite{ARW17-proceedings} seem inherently doomed
to fail to obtain results like Theorem \ref{thm:intro-main} (or even
Theorem \ref{thm:intro-edit}\footnote{See also discussion on the closely related Open Question 3 in \cite{ARW17-proceedings}}).
To understand why, we must delve into some more details of the non-deterministic
distributed PCP construction from \cite{ARW17-proceedings}. The construction
is based on a generalization of an \MA-communication protocol for
Set Disjointness by Aaronson and Wigderson \cite{AW09-algebrization}.
In the $T$-generalized protocol, (i) Merlin sends $O\left(\frac{n}{T}\log n\right)$
bits; (ii) Alice and Bob toss $O\left(\log n\right)$ coins; (iii)
Bob sends $O\left(T\log n\right)$ bits; and then (iv) Alice decides
to accept or reject. This protocol is tight up to the logarithmic
factors \cite{Klauck03-MA}, but those logarithmic factors are crucial
for our fine-grain applications!

Let $n_{\textsc{Merlin}}$ denote the length of Merlin's message,
and $n_{\textsc{Bob}}$ for Bob's message. The reduction needs to
enumerate over all of Merlin's and Bob's potential messages. To enumerate
over Merlin's messages, we simply increase the number of vectors by
a $2^{n_{\textsc{Merlin}}}$ factor; hence we must set $n_{\textsc{Merlin}}<n$.
For the above $T$-generalized protocol, this means $T=\Omega\left(\log n\right)$.
The enumeration over Bob's messages is a bit more subtle, but it roughly
corresponds to a $\left(1+2^{-n_{\textsc{Bob}}}\right)$-factor hardness
of approximation for Euclidean {\sc Bichromatic Closest Pair}. Hence
we also want to set $n_{\textsc{Bob}}$ as small as possible, and
in particular constant. But if we use the same protocol with $T=\Omega\left(\log n\right)$,
we get that $n_{\textsc{Bob}}=\Omega\left(\log^{2}n\right)$, i.e.
we don't even recover the trivial $\left(1+1/\poly\left(n\right)\right)$-hardness
of approximation. In other words, our reduction is ultimately restricted
by the PCP blowup even in the Distributed PCP framework.

\subsubsection*{AG codes}

The reason that \cite{AW09-algebrization}'s \MA-protocol incurs
an $O\left(\log n\right)$ overhead is that it uses Reed-Solomon error
correcting code (low degree polynomials), which has (relative) rate
$\Theta\left(1/\log n\right)$. Indeed, there are many error correcting
codes that achieve constant rate, but most of them don't have some
of the other nice properties of the Reed Solomon code, such as ``systemacity''
and ``polynomial closure'' (see Theorem \ref{thm:AG-codes}). In
fact, in accordance with the theme of \cite{AW09-algebrization},
their protocol seems inherently based on ``algebrization''. Building
on ideas of \cite{Meir13-IP-PSPACE,BKKMS16-linear_PCP}, we replace
the Reed Solomon code with AG codes that do satisfy the same nice
properties, and at the same time also achieve constant rate. This
allows us to reduce Merlin's and Bob's message lengths to $O\left(\frac{n}{T}\log T\right)$
and $O\left(T\log T\right)$ (Theorem \ref{thm:MA-CC}). In particular,
we can now take $T$ to be a (large) constant. 

\subsubsection*{Edit distance}

Our hardness for ANN with edit distance (Theorem \ref{thm:intro-edit})
is obtained by a black-box reduction from the hardness of ANN with
Hamming distance. The main challenge in the reduction is to rule out
small edit distance solutions obtained by ``shifting'' the vectors
with deletions. Our construction replaces each bit of the Hamming
distance instance vectors with a short random string. Now for every
$i$ such that $a_{i}=b_{i}$, the corresponding random string gadgets
are identical so the contribution to the edit distance is zero. For
$i$ such that $a_{i}\neq b_{i}$, the gadget corresponding to $a_{i}$
is matched with a uniformly random string, regardless of the shift.
Some care is required in the analysis since the behavior of edit distance
of random strings is not so well understood (even the expected edit
distance is not known).

\subsection{Related work}

\paragraph{Distributed PCP}

Most closely related to our work is \cite{ARW17-proceedings} that
introduced the Distributed PCP framework and proved SETH-based near-quadratic
hardness of approximation results for several problems in \P. In
particular, they show hardness of approximation for near-polynomial\footnote{Building on our techniques, this was later improved to polynomial
factors in \cite{Chen18}.} factors of {\sc Bichromatic Closest Pair} with the (non-metric)
similarity measures of Maximum Inner Product and Longest Common Subsequence.
More recently, \cite{AR18-IP} built on the Distributed PCP to show
that {\em deterministic} truly subquadratic-time algorithms for Longest
Common Subsequence (of two long strings) would imply new circuit lower
bounds. \cite{CLM18-FPT} use the Distributed PCP framework to rule
out Fixed Parameter Time approximation algorithms for several fundamental
problems; in particular their reduction from SETH relies on the more
efficient construction of Distributed PCP (with AG codes) from our
work. 

The latest in this line of work is \cite{Chen18}, who obtains several
closely related results: (i) building on our techniques, he obtains
more refined variants of our hardness results and the matching algorithms
(in particular, more explicit dependence on the dimension $d$); (ii)
he exhibits a hardness of approximation result (for $\left\{ \pm1\right\} $-Maximum
Inner Product) inspired by replacing \MA-communication complexity
with quantum communication complexity; (iii) he proves new hardness
of exact Euclidean {\sc Bichromatic Closest Pair} in $2^{O\left(\log^{*}n\right)}$
dimensions by replacing $\MA$-communication with $\NP\cdot\UPP$-communication;
and finally (iv) he makes significant progress on our Open Question
\ref{OQ:MA-CC}.

\paragraph{Approximate Nearest Neighbor}

The list of algorithms for Approximate Nearest Neighbor search and
Closest Pair in $\ell_{p}$ metrics is very long, e.g. \cite{SH75-mono-2d,BS76-mono,GBT84-scaling,AES91-euclidean,AM93-ANN,AMNSW94-ANN,KM95-mono,Kleinberg97-ANN,IM98,Indyk99-diameter_lb,KOR00-ANN,Indyk00-diameter,AF03-ANN,Indyk03-diameter,AI06,AMM09-ANN,AC14-ANN,AINR14,AILRS15,AR15,Valiant15,ACW16-ANN-epsilon3,ALRW16,Ahle17,Chan17-L-infty};
see also Razenshteyn's recent thesis \cite{Razenshteyn17-thesis}
for more details. As we mentioned earlier, previous impossibility
results include bounds on progress via specific algorithmic approaches
\cite{MNP07-LSH_lower_bound,OWZ14,ALRW16}, unconditional lower bounds
in the cell-probe and related models (but those hold only against
a very small number of cell probes) \cite{AIP06,PT07-probe,ACP08,PTW08,CR10-probe,PTW10,KP12,AV15,ALRW16,LPY16-probe},
and conditional hardness assuming SETH of the exact Nearest Neighbor
Search problem \cite{Wil05,AW15,DSL16-bichromatic,Wil18}. For edit
distance Approximate Nearest Neighbor search, after a sequence of
improving embedding-based algorithms \cite{Indyk04-product,bar2004approximating,batu2006oblivious,ostrovsky2007low},
the state of the art for truly subquadratic time algorithms in $d$
dimensions is a $2^{O(\sqrt{\log{d}\log\log{d}})}$-approximation
\cite{ostrovsky2007low}. For embedding into specific norms, some
lower bounds were proven, e.g. \cite{ADG+03,KN05,AK07,AJP10}. SETH-hardness
for the exact problem is also known \cite{ABV15a,BI15}.

\paragraph{AG codes}

Algebraic Geometry (AG) codes generalize Reed Solomon and Reed Muller
codes. The first construction is due to Goppa \cite{Goppa81}, but
we use an efficiently decodable construction from \cite{SAKSD01-AG-codes}.
AG codes have received significant attention from computer scientists
(e.g. \cite{GS99-AG-codes,GP08-AG-codes,BT13-AG-codes,CT13a-AG-codes,GX13-AG-codes,GX14-AG-codes,HRW17-AG-codes}).
In particular, Meir \cite{Meir13-IP-PSPACE} observed that ``algebrization-like''
construction (in particular, the proof of $\IP=\PSPACE$) can be obtained
with any error correcting codes that satisfy properties that he calls
multiplicativity (polynomial closure) and systemacity. Ben-Sasson
et al. \cite{BKKMS16-linear_PCP} observed that AG codes satisfy those
properties and at the same time achieve a constant rate. They used
AG codes to construct PCP of linear size but polynomial query complexity.
Finally, \cite{BCGRS16-IOP} used AG codes to construct Interactive
Oracle Proofs of linear size, constant query complexity, and a constant
number of rounds. We note that our construction is the first to obtain
qualitatively new hardness results from AG codes (and our ideas have
already inspired even more hardness results in followup work \cite{CLM18-FPT,Chen18}!).

\subsection{Discussion and open problems}

While our application of AG codes to the \MA-communication protocol
is very simple and easily yields strong results like Theorem \ref{thm:intro-main},
there are a few closely related applications where we are unable to
shave the logarithmic factors using AG codes.

\paragraph{\MA-communication complexity of Set Disjointness}

\cite{AW09-algebrization}'s original protocol gives an $O\left(\sqrt{n}\log n\right)$
\MA-communication protocol for Set Disjointness. Naively, one would
expect that it should be possible to obtain an $O\left(\sqrt{n}\right)$
protocol using AG codes, which would be tight by \cite{Klauck03-MA}.
However, there is a step in the protocol that requires verifying that
a vector in $\left\{ 0,1\right\} ^{\sqrt{n}}$ is all zeros by looking
at its sum; we only know how to perform this step over a large field.
In Theorem \ref{thm:AMA-CC} we use related ideas to obtain an $O\left(\sqrt{n}\right)$
\AMA-communication complexity, but we don't know how to do that for
\MA-communication. Note also that \cite{AW09-algebrization}'s protocol
also solves the more general Inner Product problem for the same communication
complexity. In an earlier version of this paper, we asked whether
the \MA-communication complexity of (Set Disjointness / Inner Product)
is $\Theta\left(\sqrt{n}\right)$ or $\Theta\left(\sqrt{n}\log n\right)$?
This was partially answered by \cite{Chen18} who gave an $O\left(\sqrt{n\log n\log\log n}\right)$
protocol using AG codes.

\begin{oq}[\MA-communication complexity]\label{OQ:MA-CC}What are
the \MA-communication complexities of the Set Disjointness and Inner
Product problems (known to be between $\Omega\left(\sqrt{n}\right)$
and $O\left(\sqrt{n\log n\log\log n}\right)$)?

\end{oq}

\paragraph{\IP-communication complexity of Set Disjointness}

\cite{AW09-algebrization}'s protocol also generalizes to an $O\left(\log n\log\log n\right)$
\IP-protocol. The \IP-protocol uses $O\left(\log n\right)$ rounds,
where at each round a Reed Solomon code over field size $O\left(\log n\right)$
is used. One may hope to obtain an $O\left(\log n\right)$ bits protocol
using AG codes. But it is crucial that the failure probability in
each round is $1-O\left(1/\log n\right)$, whereas AG codes of constant
rate cannot obtain distance $1-o\left(1\right)$. We note that this
\IP~protocol is used in \cite{AR18-IP} in the context of Distributed
PCP, and reducing the communication to $O\left(\log n\right)$ would
imply stronger conditional circuit lower bounds.

\begin{oq}[\IP-communication complexity]Is there an $O\left(\log n\right)$
\IP-communication complexity protocol for Set Disjointness?

\end{oq}

\paragraph{Exact {\sc Bichromatic Closest Pair} with short vectors}

Williams \cite{Wil18} recently used a construction closely related
to \cite{AW09-algebrization}'s protocol to prove that, assuming SETH,
exact {\sc Bichromatic Closest Pair} with Euclidean distance requires
near-quadratic time even for vectors of dimension $O\left(\log^{2}\log N\right)$.
(He also proves similar results for {\sc Furthest Pair} and {\sc Orthogonal Vectors}.)
He poses an open problem to obtain a similar result for constant dimension.
Again, the bottleneck is the log-of-field-size factor in the communication
protocol, which one may hope to overcome by using constant size fields.
Recent work by Chen \cite{Chen18} significantly reduced the dimension
to $2^{O\left(\log^{*}N\right)}$; but achieving constant dimension
remains open.

\subsubsection{Looser approximations of {\sc Bichromatic Closest Pair}}

Our techniques can only rule out very good approximations of {\sc Bichromatic Closest Pair}
in strongly-subquadratic time. Indeed for weaker approximation factors
there are strongly subquadratic time algorithms. But can we rule out
near-linear time algorithms for, say, any constant approximation factor?
Even better, can we obtain tight conditional result for every approximation
factor, as \cite{ALRW16} do for data-dependent LSH? Below we discuss
some of the technical barriers that seem to arise toward proving such
results. We stress that the discussion below is highly informal and
should not be used to draw any definite conclusions.

\paragraph{Finer-grain PCP}

Reductions from SETH to problems in \P~are typically classified as
``fine grain complexity'' because they provide a very refined understanding
of the running time of problems like {\sc Bichromatic Closest Pair},
up to sub-polynomial factors. This means that we have to nail the
``correct'' size of the reduction almost exactly \textemdash{} to
within sub-polynomial factors. The PCP blowup we described earlier
arises because the reduction size is exponential in the length of
the PCP; so our PCP length should be optimal up to sub-{\em constant}
factors! 

The Distributed PCP can be seen as a technique for decoupling the
size of the reduction into two parts: the exact part that reduces
$k$-SAT to some problem in \P~(e.g. {\sc Bichromatic Closest Pair}),
and the PCP-like part that is responsible for hardness or approximation
gap amplification. In other words, the size of the new instance $N_{\textsc{Approx}}$
can roughly be described as the product of the size of the exact reduction,
$N\approx2^{n/2}$, and the gap amplification factor, $N_{\textsc{Gap}}$:
\[
N_{\textsc{Approx}}\approx N\cdot N_{\textsc{Gap}}.
\]
The first (exact) part is still required to have the correct size,
but if the ``correct'' blowup of gap amplification ($N_{\textsc{Gap}}$)
is tiny, then we can afford some redundancy. For example, in \cite{ARW17-proceedings}
the correct blowup for gap amplification was $N_{\textsc{Gap}}=N^{o\left(1\right)}$,
so we could afford to use a much less efficient (losing super-polynomial
factors) gap amplification $\widehat{N_{\textsc{Gap}}}=\left(N_{\textsc{Gap}}\right)^{\omega\left(1\right)}$,
while barely affecting the reduction size: $\widehat{N_{\textsc{Gap}}}=N^{o\left(1\right)}$.
In the present paper, the correct gap blowup for proving $\left(1+\varepsilon\right)$-inapproximability
is $N_{\textsc{Gap}}=N^{f\left(\varepsilon\right)}$; using AG codes
we can obtain $\widehat{N_{\textsc{Gap}}}=N^{g\left(\varepsilon\right)}$
blowup, which is still OK when $\varepsilon$ is sufficiently small.
In contrast, by \cite{AR15}, a $2$-approximation can be obtained
in time 
\[
T\left(N\right)\approx\left(N_{\textsc{Approx}}\right)^{8/7}\approx\left(N\cdot N_{\textsc{Gap}}\right)^{8/7};
\]
plugging in an $T\left(N\right)\approx N^{2}$ lower bound from SETH
and solving for $N_{\textsc{Gap}}$, we estimate that the correct
blowup for a factor $2$ gap has to be at least $N_{\textsc{Gap}}\approx N^{6/7}$.
Therefore if we don't get the correct exponent to within a $7/6$-factor,
we cannot keep the total instance size under $N^{2}$, and hence cannot
obtain any super-linear lower bounds on the running time. Even using
AG codes, it seems that we lack the techniques for obtaining such
fine-grained gap amplification. In particular, all PCP constructions
(including Dinur's combinatorial PCP \cite{Dinur07-PCP}) use error
correcting codes, which are inherently redundant.

\paragraph{Triangle inequality}

This barrier is specific to obtaining hardness of approximation for
factor $3$ or greater. Consider a naive gadget reduction where we
construct a vector for each half assignment. Let $\alpha_{1},\alpha_{2}\in\left\{ 0,1\right\} ^{n/2}$
be partial assignments to the first half of the variables, and $\beta_{1},\beta_{2}\in\left\{ 0,1\right\} ^{n/2}$
for the second half. Suppose that $\left(\alpha_{1};\beta_{1}\right),\left(\alpha_{2};\beta_{1}\right),\left(\alpha_{2};\beta_{2}\right)$
satisfy the formula, but $\left(\alpha_{1},\beta_{2}\right)$ does
not. Let $a^{\alpha_{1}},a^{\alpha_{2}},b^{\beta_{1}},b^{\beta_{2}}$
be the corresponding vectors. Then, if our reduction has completeness
$c$ and soundness $s$, we would like to have
\[
\left\Vert a^{\alpha_{1}}-b^{\beta_{2}}\right\Vert \geq s\geq3c\geq\left\Vert a^{\alpha_{1}}-b^{\beta_{1}}\right\Vert +\left\Vert a^{\alpha_{2}}-b^{\beta_{2}}\right\Vert +\left\Vert a^{\alpha_{2}}-b^{\beta_{1}}\right\Vert .
\]
But that would violate the triangle inequality. Note that this restricts
our ability to prove stronger hardness of approximation even for more
complicated metrics like edit distance. It is also important to remark
that the reductions based on Distributed PCP (including the ones in
this paper) do not exactly fall into this naive gadget reduction framework;
nevertheless it is not at all clear that they can overcome this obstacle.

\begin{oq}[$3$-approximation] Prove that, assuming SETH and for some
constant $\varepsilon>0$, approximating {\sc Bichromatic Closest Pair}
with Euclidean metric to within factor $3$ requires time $\Omega\left(N^{1+\varepsilon}\right)$.

\end{oq}

\section{Preliminaries}

\subsection{Complexity assumptions}
\begin{conjecture}
[Strong Exponential Time Hypothesis (SETH)~\cite{IP01-SETH}]\label{conj:SETH}
For any $\epsilon>0$, there exists $k=k\left(\epsilon\right)$ such
that {\sc $k$-SAT} on $n$ variables cannot be solved in time $O\left(2^{\left(1-\epsilon\right)n}\right)$.
\end{conjecture}

SETH is in particular known to imply the Orthogonal Vectors Conjecture
(OVC) \cite{Wil05}, which postulates a quadratic-time hardness for
the following {\sc Orthogonal Vectors} problem:
\begin{problem}
[{\sc Orthogonal Vectors}] Given two sets $A,B\subset\left\{ 0,1\right\} ^{m}$,
decide whether there is a pair $\left(a,b\right)\in A\times B$ such
that $a\cdot b=0$.
\end{problem}

\begin{conjecture}
[Orthogonal Vectors Conjecture (OVC)]\label{conj:OVC} For every
$\delta>0$ there exists $c=c\left(\delta\right)$ such that given
two sets $A,B\subset\left\{ 0,1\right\} ^{m}$ of cardinality $N$,
where $m=c\log N$, deciding if there is a pair $\left(a,b\right)\in A\times B$
such that $a\cdot b=0$ cannot be solved in time $O\left(N^{2-\delta}\right)$.
\end{conjecture}

\subsection{AG codes}

Algebraic Geometry (AG) codes generalize Reed Solomon and Reed Muller
codes. The first construction is due to Goppa \cite{Goppa81}, but
we use an efficiently decodable construction from \cite{SAKSD01-AG-codes}.
The codewords of Reed Muller codes are low-degree multivariate polynomials.
AG codewords are multivariate rational functions (quotients of polynomials).
In particular, like Reed Solomon or Reed Muller codes, we can add
codewords (for free) and multiply codewords (with some loss in the
distance). Here, we only use the following basic properties of AG
codes; for a thorough introduction, we refer the reader to \cite{Stichtenoth09-textbook}.
\begin{thm}
[\cite{SAKSD01-AG-codes}]\label{thm:AG-codes} There exists a constant
$q_{0}\in\mathbb{N}$, such that for every prime $q\geq q_{0}$, there
exist two code families ${\cal C}\triangleq\left\{ C_{n}\right\} $
and ${\cal C}^{'}\triangleq\left\{ C_{n}^{'}\right\} $ whose codewords
are given by functions $w:\mathcal{R}_{n}\rightarrow\mathbb{F}_{q^{2}}$,
for $\mathcal{R}_{n}\subset\mathbb{F}_{q^{2}}^{O\left(\log n\right)}$.
Furthermore, those code families satisfy the following properties:
\begin{description}
\item [{Systematicity}] There exists a subset $\mathcal{S}_{n}\subset\mathcal{R}_{n}$
of cardinality $\left|\mathcal{S}_{n}\right|=\Theta\left(n\right)$,
such that for any assignment $x:\mathcal{S}_{n}\rightarrow\mathbb{F}_{q^{2}}$,
there exists a codeword $w\in\mathcal{C}$ such that $w\mid_{S_{n}}=x$. 
\item [{Polynomial Closure}] ${\cal C}$ and ${\cal C}^{'}$ are linear
codes; furthermore, for every $w_{1},w_{2}\in\mathcal{C}$, there
exists $w^{'}\in{\cal C}^{'}$ such that for every $i\in\mathcal{R}_{n}$,
$w^{'}(i)=w_{1}(i)\cdot w_{2}(i)$.
\item [{Parameters}] Both codes have (relative) rate at $0.1$ and (relative)
distance at least $0.1$.
\item [{Efficiency}] Both codes can be encoded and checked in $\poly\left(n\right)$
time.
\end{description}
\end{thm}

\section{communication complexity}

\subsection{\MA~communication complexity }
\begin{thm}
\label{thm:MA-CC}The following holds for every $T\in\left[2,m\right]$.
There is a computationally efficient $\MA$-communication protocol
for Set Disjointness over universe $\left[m\right]$ where:

\begin{enumerate}
\item Merlin sends Alice $O\left(\frac{m\log T}{T}\right)$ bits.
\item Alice and Bob toss $O\left(\log m\right)$ coins.
\item Bob sends Alice $O\left(T\log T\right)$ bits.
\item Alice returns {\em Accept} or {\em Reject}.
\end{enumerate}
If the sets are disjoint, Alice always accepts. Otherwise, she accepts
with probability at most $1/2$. 
\end{thm}

\begin{proof}
Assume wlog that $T$ divides $m$. We partition the universe into
$T$ disjoint sets of size $m/T$: $\left[m\right]=U^{1}\cup\dots\cup U^{T}$.
Let $\alpha,\beta\subseteq\left[m\right]$ denote Alice and Bob's
respective inputs; for $t\in\left[T\right]$, let $\alpha^{t}\triangleq\alpha\cap U^{t}$,
and define $\beta^{t}$ analogously. 

Let $C$ be an algebraic geometry code over field $\mathbb{F}_{q^{2}}$
with characteristic at least $q\geq T$, as promised by Theorem \ref{thm:AG-codes};
let $\rho_{C}$, $\delta_{C}$, and $n_{C}=\frac{m}{T\cdot\rho_{C}}=O\left(m/T\right)$
denote the relative rate, relative distance, and the length of its
codewords. For $t\in\left[T\right]$, let $C\left(\alpha^{t}\right),C\left(\beta^{t}\right)$
denote the encodings of $\alpha^{t}$ and $\beta^{t}$. Then, by the
multiplicative (polynomial closure) property, their entry-wise product
$\mu^{t}$ (i.e. ${\mu^{t}}_{i}\triangleq\left[C\left(\alpha^{t}\right)\right]_{i}\cdot\left[C\left(\beta^{t}\right)\right]_{i}$)
is a codeword in $C^{'}$. By linearity of $C^{'}$, the entry-wise
sum of the $\mu^{t}$'s ($\mu_{i}\triangleq\sum_{t=1}^{T}{\mu^{t}}_{i}$)
is also a codeword of $C^{'}$.

Since $C$ is a systematic code, we have that for each $i\in\left[n/T\right]$,
$\left[C\left(\alpha^{t}\right)\right]_{i},\left[C\left(\beta^{t}\right)\right]_{i}\in\left\{ 0,1\right\} $,
and therefore also ${\mu^{t}}_{i}\in\left\{ 0,1\right\} $. Observe
that the sets are disjoint iff ${\mu^{t}}_{i}=0$ for all $i\in\left[m/T\right]$
and $t\in\left[T\right]$. Because the characteristic is at least
$T$, this is equivalent to requiring that $\mu_{i}=0$ for all $i\in\left[m/T\right]$. 

Now the protocol proceeds as follows:

\begin{enumerate}
\item Merlin sends Alice $\hat{\mu}$, which is allegedly the encoding of
$\mu$. 
\item Alice and Bob pick a random $i^{*}\in\left[n_{C}\right]$.
\item Bob sends Alice $\left[C\left(\beta^{t}\right)\right]_{i^{*}}$, for
all $t\in\left[T\right]$.
\item Alice accepts iff all of the following hold:

\begin{enumerate}
\item $\hat{\mu}$ is a codeword in $C^{'}$;
\item $\hat{\mu}_{i^{*}}=\sum_{t=1}^{T}\left[C\left(\alpha^{t}\right)\right]_{i^{*}}\cdot\left[C\left(\beta^{t}\right)\right]_{i^{*}}$;
and
\item $\hat{\mu}_{i}=0$ for all $i\in\left[m/T\right]$.
\end{enumerate}
\end{enumerate}
First, we observe that Merlin sends $n_{C}\cdot\lceil\log T\rceil=O\left(\log T\cdot m/T\right)$
bits, and Bob sends $T\cdot\lceil\log T\rceil$ bits, so the communication
complexity is as promised in the theorem statement.

If Alice accepts with nonzero probability, then $\hat{\mu}$ is a
codeword of $C^{'}$. Therefore if Alice accepts with probability
greater than $1-\delta_{C^{'}}$ (where $\delta_{C^{'}}\geq0.1$ is
the relative distance of $C'$) then $\hat{\mu}$ is also equal to
the true $\mu$. Then this also implies that $\mu_{i}=0$ for all
$i\in\left[m/T\right]$. Hence by our earlier observation, the sets
are disjoint.

Finally, repeat Steps 2-4 of the protocol (in parallel) a constant
number of times to obtain soundness $0.5$.
\end{proof}

\subsection{\AMA~communication complexity}
\begin{thm}
\label{thm:AMA-CC}The following holds for every $T\in\left[2,m\right]$.
There is a computationally efficient $\AMA$-communication protocol
for Set Disjointness over universe $\left[m\right]$ where:

\begin{enumerate}
\item Alice and Bob toss $O\left(\log m\right)$ coins, and send them to
Merlin.
\item Merlin sends Alice $O\left(m/T\right)$ bits.
\item Alice and Bob toss $O\left(\log m\right)$ coins.
\item Bob sends Alice $O\left(T\right)$ bits.
\item Alice returns {\em Accept} or {\em Reject}.
\end{enumerate}
If the sets are disjoint, Alice always accepts. Otherwise, she accepts
with probability at most $1/2$.
\end{thm}

\begin{proof}
The set up is similar to the proof of Theorem \ref{thm:MA-CC}, but
we use error correcting codes $C,C^{'}$ over a field $\mathbb{F}_{q^{2}}$
of constant size (that does not depend on $T$ or $m$).

Partition the set $\left[m\right]$ into $T$ parts, and let $\alpha^{t},\beta^{t}$
be the restriction of $\alpha,\beta$ (respectively) to the $t$-th
part. Let $C\left(\alpha^{t}\right),C\left(\beta^{t}\right)$ denote
the encodings of $\alpha^{t},\beta^{t}$, respectively, and let $\mu^{t}$
denote their entrywise product (${\mu^{t}}_{i}\triangleq\left[C\left(\alpha^{t}\right)\right]_{i}\cdot\left[C\left(\beta^{t}\right)\right]_{i}$).

At the first step of the protocol, Alice and Bob pick a random subset
$S\subseteq\left[T\right]$. For now, let us assume that $S$ is chosen
uniformly at random (note that this already suffices to show $O\left(\sqrt{n}\right)$
\AMA-communication complexity). We will later reduce the number of
coins tossed in this step using a standard argument due to Newman
\cite{Newman91-logn-coins}.

We now define $\mu\left(S\right)$ as the entrywise sum of $\mu^{t},$
taken only over $t\in S$: 
\[
\mu_{i}\left(S\right)\triangleq\sum_{t\in S}{\mu^{t}}_{i}.
\]

Observe that the sets are disjoint iff ${\mu^{t}}_{i}=0$ for all
$i\in\left[m/T\right]$ and $t\in\left[T\right]$. If the sets are
indeed disjoint, we also have that $\mu_{i}=0$ for all $i\in\left[m/T\right]$.
Otherwise, there exists $i\in\left[m/T\right]$ and $t\in\left[T\right]$
such that ${\mu^{t}}_{i}=1$. Fix some choice of $S\setminus\left\{ t\right\} $,
and hence also the value of $\sum_{t'\in S\setminus\left\{ t\right\} }{\mu^{t'}}_{i}$.
Conditioned on those, $t\in S$ with probability exactly $1/2$. Therefore,
for any value of $\sum_{t'\in S\setminus\left\{ t\right\} }{\mu^{t'}}_{i}$,
we have that the sum with $t$ is nonzero with probability at least
$1/2$.

Now the protocol proceeds as follows:
\begin{enumerate}
\item Alice and Bob pick $S\subseteq\left[T\right]$ at random, and send
it to Merlin
\item Merlin sends Alice $\hat{\mu}\left(S\right)$, which is allegedly
the encoding of $\mu\left(S\right)$. 
\item Alice and Bob pick a random $i^{*}\in\left[n_{C}\right]$.
\item Bob sends Alice $\left[C\left(\beta^{t}\right)\right]_{i^{*}}$, for
all $t\in\left[T\right]$.
\item Alice accepts iff all of the following hold:

\begin{enumerate}
\item $\hat{\mu}$ is a codeword in $C^{'}$;
\item $\hat{\mu}_{i^{*}}\left(S\right)=\sum_{t\in S}\left[C\left(\alpha^{t}\right)\right]_{i^{*}}\cdot\left[C\left(\beta^{t}\right)\right]_{i^{*}}$;
and
\item $\hat{\mu}_{i}\left(S\right)=0$ for all $i\in\left[m/T\right]$.
\end{enumerate}
\end{enumerate}
If Alice accepts with nonzero probability, then $\hat{\mu}$ is a
codeword of $C^{'}$. Therefore if Alice accepts with probability
greater than $1-\delta_{C^{'}}$ (where $\delta_{C^{'}}\geq0.1$ is
the relative distance of $C'$) then $\hat{\mu}$ is also equal to
the true $\mu$. Then this also implies that $\mu_{i}=0$ for all
$i\in\left[m/T\right]$. Hence by our earlier observation, if Alice
accepts with probability greater than $1-\delta_{C^{'}}/2$, the sets
must be disjoint.

To amplify the soundness to $1/2$, we need to repeat the protocol.
For fixed non-disjoint input $\alpha,\beta$, we say that Step 1 of
the protocol fails if $\mu_{i}=0$ for all $i\in\left[m/T\right]$.
We repeat this step of the protocol twice (in parallel) so the probability
that both repetitions fail is only $1/4$. Notice that it is safe
to repeat this step in parallel (i.e. we do not need to invoke a ``parallel
repetition'' theorem) because the failure probability of this step
does not depend on Merlin's actions. Similarly, we repeat Steps 3-5
of the protocol, in parallel to amplify the overall soundness of the
protocol to $1/2$.

Finally, instead of sampling $S$ uniformly at random, Alice, Bob
and Merlin can agree in advance on a candidate family ${\cal S}$
of sets $S^{j}$, chosen uniformly at random. For any fixed input
$\alpha,\beta$, the probability that Step 1 fails on significantly
more than half of the $S^{j}$ is $2^{-O\left(\left|{\cal S}\right|\right)}$.
Setting $\left|{\cal S}\right|=\Theta\left(m\right)$ allows us to
take a union bound over all choices of $\alpha,\beta$. Therefore
Alice and Bob only need $\log\left|{\cal S}\right|=O\left(\log m\right)$
coins to choose a random set $S^{j}\in{\cal S}.$
\end{proof}

\section{Approximate Closest Pair}
\begin{thm}
[Detailed version of Theorem~\ref{thm:intro-main}]\label{thm:main}Unless
SETH and OVC are false, the following holds for any $\ell_{p}$ metric,
for constant $p$ (including Hamming, Manhattan, and Euclidean distance):
for every $\delta>0$ there exists an $\epsilon=\epsilon\left(\delta,p\right)$
such that given a set $A,B\subset\left\{ 0,1\right\} ^{d}$ of $N$
vectors (where $d=O\left(\log N\right)$), computing a $\left(1+\epsilon\right)$-approximation
to {\sc Bichromatic Closest Pair} requires time $\Omega\left(N^{2-\delta}\right)$. 
\end{thm}

\begin{proof}
We consider the \MA-protocol from Theorem \ref{thm:MA-CC}, instantiated
with parameter $T=T\left(\epsilon\right)=O\left(\frac{\log\frac{1}{\epsilon}}{\log\log\frac{1}{\epsilon}}\right)$.
Let $T'=2^{O\left(\log T\cdot T\right)}$ denote the number of different
possible messages sent by Bob in the protocol. We set $T$ so that
$T'=O\left(1/\epsilon\right)$. 

Our reduction takes as input an instance $\left(A_{\textsc{OV}},B_{\textsc{OV}}\right)$
of {\sc Orthogonal Vectors} over $\left\{ 0,1\right\} ^{m}$. For
each vector $\beta\in B_{\textsc{OV}}$, we construct a new vector
$\tilde{b}^{\beta}\in\left\{ 0,1\right\} ^{T'\times m}$ by setting
$\left[\tilde{b}^{\beta}\right]_{i,j}\triangleq1$ iff Bob sends message
$i$ on input $\beta'$ and randomness $j$. For each Merlin-message
$\mu\in\left\{ 0,1\right\} ^{O\left(\log T\cdot m/T\right)}$ and
each vector $\alpha\in A_{\textsc{OV}}$, we construct a new vector
$\tilde{a}^{\mu,\alpha}\in\left\{ 0,1\right\} ^{T'\times m}$ as follows:
$\left[\tilde{a}^{\mu,\alpha}\right]_{i,j}\triangleq1$ iff Alice
accepts on input $\alpha$, message $\mu$ from Merlin, message $i$
from Bob, and randomness $j$. Notice also that the inner product
of two vectors $\tilde{a}^{\mu,\alpha}\cdot\tilde{b}^{\beta}$ is
exactly proportional to the probability that Alice and Bob accept
on inputs $\alpha,\beta$ and message $\mu$ from Merlin. In particular,
if $\alpha$ and $\beta$ are not disjoint/orthogonal, the inner product
is at most $\tilde{a}^{\mu,\alpha}\cdot\tilde{b}^{\beta}\leq m/2$.
Otherwise, there exists a $\mu$ such that $\tilde{a}^{\mu,\alpha}\cdot\tilde{b}^{\beta}=m$.

To argue about distances we have to make one small modification. We
replace vector $\tilde{a}\in\left\{ 0,1\right\} ^{T'\times m}$ by
vector $a\in\left\{ 0,1\right\} ^{T'\times m\times2}$ constructed
as follows: $a_{i,j,1}\triangleq\tilde{a}_{i,j}$ and $a_{i,j,2}\triangleq1-\tilde{a}_{i,j}$.
This guarantees that the number of $1$'s in every $a$-vector is
the same ($T'\cdot m$). The number of $1$'s in $b$-vectors was
already the same ($m$), so we can just define $b_{i,j,1}\triangleq\tilde{b}_{i,j}$
and $b_{i,j,2}\triangleq0$. We now have the following
\begin{align*}
\Pr_{\left(i,j,b\right)\in\left[T\right]'\times\left[m\right]\times\left[2\right]}\left[\left[\tilde{a}^{\mu,\alpha}\right]_{i,j,b}\neq\left[\tilde{b}^{\beta}\right]_{i,j,b}\right] & =\Pr\left[\left[\tilde{a}^{\mu,\alpha}\right]_{i,j,b}=1\right]+\Pr\left[\left[\tilde{b}^{\beta}\right]_{i,j,b}=1\right]-2\left(\frac{\tilde{a}^{\mu,\alpha}\cdot\tilde{b}^{\beta}}{T'\cdot m\cdot2}\right)\\
 & =\frac{1}{2}+\frac{1}{2T'}-\frac{1}{T'}\Pr\left[\text{Alice accepts}\right].
\end{align*}
In particular, for every $p>0$ we have that
\[
\min_{\mu}\left\Vert \tilde{a}^{\mu,\alpha}-\tilde{b}^{\beta}\right\Vert _{p}^{p}=\begin{cases}
m\left(T'-1\right) & \alpha\cap\beta=\emptyset\\
\geq mT' & \text{otherwise}
\end{cases}.
\]

For any constant $p$, taking the $\left(1/p\right)$-th power we
have that {\sc Orthogonal Vectors} reduces to approximating {\sc Bichromatic Closest Pair}
in $\ell_{p}$ norm to within a factor of 
\[
\left(1+\frac{1}{T'-1}\right)^{1/p}\geq1+\Omega_{p}\left(\frac{1}{T'}\right).
\]
The result for Hamming distance follows from Manhattan distance since
all the entries are in $\left\{ 0,1\right\} $.

Finally, observe that our reduction increases the number of vectors
by a factor of $2^{O\left(\log T\cdot m/T\right)}=2^{O\left(m\cdot\frac{\log^{2}\log\frac{1}{\epsilon}}{\log\frac{1}{\epsilon}}\right)}$.
Hence in total, if there are $N_{\textsc{OV}}=2^{m/c}$ vectors in
the OVC instance, our new instance of {\sc Bichromatic Closest Pair}
has $N=2^{m\left(1/c+O\left(\frac{\log^{2}\log\frac{1}{\epsilon}}{\log\frac{1}{\epsilon}}\right)\right)}$
vectors. In particular, approximating {\sc Bichromatic Closest Pair}
to within $1+\epsilon$ is as hard as solving the original instance
of {\sc Orthogonal Vectors}, which by OVC requires time 
\begin{align}
\left(\left|A_{\textsc{OV}}\right|+\left|B_{\textsc{OV}}\right|\right)^{2-\delta_{\textsc{OV}}} & =\left(2^{m/c}\right)^{2-\delta_{\textsc{OV}}}\nonumber \\
 & =N^{2}/2^{m\left(\delta_{\textsc{OV}}/c-O\left(\frac{\log^{2}\log\frac{1}{\epsilon}}{\log\frac{1}{\epsilon}}\right)\right)}\nonumber \\
 & =N^{2-\delta},\label{eq:delta-epsilon}
\end{align}
where the last equality follows by choosing $\delta=O\left(\delta_{\textsc{OV}}+\frac{c\log^{2}\log\frac{1}{\epsilon}}{\log\frac{1}{\epsilon}}\right)$,
where $\delta_{\textsc{OV}}$ and $c=c\left(\delta_{\textsc{OV}}\right)$
are as in the OVC. 

Finally, notice that the vectors we construct have dimension $2T'\cdot m=O\left(m\right)=O\left(\log N\right)$.
\end{proof}

\subsection{Edit distance}

Below we restate and prove our hardness for {\sc Bichromatic Closest Pair}
with edit distance. In the proof, we use a somewhat non-standard (but
equivalent) notion of edit distance, where we think of characters
from both strings being deleted or substituted (as opposed to inserting,
deleting, and substituting characters from one string).
\begin{thm}
[Detailed version of Theorem~\ref{thm:intro-edit}]\label{thm:edit}Unless
SETH and OVC are false, the following holds for edit distance: for
every $\delta>0$ there exists an $\epsilon=\epsilon\left(\delta\right)$
such that given a set $A,B\subset\left\{ 0,1\right\} ^{d}$ of $N$
vectors (where $d=O\left(\log N\log\log N\right)$), computing a $\left(1+\epsilon\right)$-approximation
to {\sc Bichromatic Closest Pair} requires time $\Omega\left(N^{2-\delta}\right)$. 
\end{thm}

\begin{proof}
We first provide a randomized construction. We will later describe
how to efficiently derandomize it. 

We begin with a hard instance $\left(A_{\textsc{H}},B_{\textsc{H}}\right)$
of Hamming distance {\sc Bichromatic Closest Pair} as guaranteed
by Theorem \ref{thm:main}. Let $d_{\textsc{H}}=O\left(\log N\right)$
be the dimension of the vectors in $\left(A_{\textsc{H}},B_{\textsc{H}}\right)$.
We draw $2d_{\textsc{H}}$ binary random strings of dimension $d'=O\left(\log d_{\textsc{H}}\right)$
uniformly at random, denoted $s_{0}^{i},s_{1}^{i}\in\left\{ 0,1\right\} ^{d'}$,
$\forall i\in\left[d_{\textsc{H}}\right]$. By Claim \ref{claim:expectation},
for any two random strings $s_{c}^{i},s_{c'}^{i'}$, the expectation
of their edit distance $\lambda_{2,d'}^{\ED}\triangleq\E\left[\ED\left(s_{c}^{i},s_{c'}^{i'}\right)\right]$
is bounded from below by $\Omega\left(d'\right)$. By Claim \ref{claim:concentration},
their edit distance $\ED\left(s_{c}^{i},s_{c'}^{i'}\right)$ also
concentrates around its expectation to within $\pm o\left(d'\right)$,
with probability $1-1/\poly\left(d_{\textsc{H}}\right)$. Furthermore,
by applying the same claim with union bound, if we consider any choice
$c_{1},\dots,c_{d_{\textsc{H}}}$, and any contiguous substring of
$s_{c_{1}}^{1}\circ\cdots\circ s_{c_{d_{\textsc{H}}}}^{d_{\textsc{H}}}$,
it's distance from any string $s_{1-c_{i}}^{i}$ is within $\left(1\pm o\left(1\right)\right)$
factor of the expected edit distance of two strings of those lengths
(w.h.p.). 

Our reduction is constructed as follows. For each vector $u_{\textsc{H}}\in A_{\textsc{H}}\cup B_{\textsc{H}}$
we construct a vector $u\in\left\{ 0,1\right\} ^{d}$, for $d\triangleq d_{\textsc{H}}\cdot d'$,
by replacing the $i$-th bit of $u_{\textsc{H}}$ with one of $s_{0}^{i},s_{1}^{i}$.
Let $\left[u_{\textsc{H}}\right]_{i}$ denote the $i$-th bit of $u_{\textsc{H}}$;
then define $u^{i}\triangleq\text{\ensuremath{s_{\left[u_{\textsc{H}}\right]_{i}}^{i}}},$
and also 
\[
u\triangleq u^{1}\circ\cdots\circ u^{d_{\textsc{H}}}=s_{\left[u_{\textsc{H}}\right]_{1}}^{1}\circ\cdots\circ s_{\left[u_{\textsc{H}}\right]_{d_{\textsc{H}}}}^{d_{\textsc{H}}}.
\]

\subsubsection*{Analysis}

We now argue that for any $u_{\textsc{H}},v_{\textsc{H}}$, we have
that
\[
\ED\left(u,v\right)=\lambda_{2,d'}^{\ED}\cdot\Delta_{\textsc{Hamming}}\left(u_{\textsc{H}},v_{\textsc{H}}\right)\pm o\left(d\right).
\]
By matching every pair of $i$-th strings, it is clear that $\ED\left(u,v\right)\leq\lambda_{2,d'}^{\ED}\cdot\Delta_{\textsc{Hamming}}\left(u_{\textsc{H}},v_{\textsc{H}}\right)+o\left(d\right)$.
To see the other direction, it remains to show that there doesn't
exist a much better way to edit $u$ and $v$. Recall the partition
of $u$ into the $d_{\text{H}}$ substrings $u^{i},\dots,u^{d_{\textsc{H}}}$.
Fix an optimal way to edit $u$ and $v$, and partition $v$ into
$d_{\text{H}}$ contiguous substrings $\hat{v}^{1},\dots,\hat{v}^{d_{\textsc{H}}}$,
where all the characters in each $\hat{v}^{i}$ are either deleted
or matched to a character from $u^{i}$. We now have that 
\begin{equation}
\ED\left(u,v\right)=\sum_{i=1}^{d_{\textsc{H}}}\ED\left(u^{i},\hat{v}^{i}\right).\label{eq:sum}
\end{equation}
Let $G\subseteq\left[d_{\textsc{H}}\right]$ denote the set of coordinates
on which $u_{\textsc{H}},v_{\textsc{H}}$ agree (notice that $\Delta_{\textsc{Hamming}}\left(u_{\textsc{H}},v_{\textsc{H}}\right)=d_{\begin{tabular}{c}
 H\end{tabular}}-\left|G\right|$). We lower bound the sum on the RHS of (\ref{eq:sum}) by considering
$i\in G$ and $i\notin G$ separately. 
\begin{itemize}
\item For $i\in G$, we lower bound the contribution to the edit distance
by the difference in length, namely
\begin{equation}
\ED\left(u^{i},\hat{v}^{i}\right)\geq\left|\left|u^{i}\right|-\left|\hat{v}^{i}\right|\right|,\label{eq:ui=00003Dvi}
\end{equation}
 where we use $\left|\hat{v}^{i}\right|$ to denote the length of
$\hat{v}^{i}$ (in particular, $\left|u^{i}\right|=d'$ for all $i$).
\item Consider the sum of the edit distances over $i\notin G$. Let us first
replace $u^{i},\hat{v}^{i}$ with freshly drawn random $x^{i},y^{i}$
of same lengths. By concentration, we have that
\begin{align*}
\sum_{i\notin G}\ED\left(u^{i},\hat{v}^{i}\right) & \geq\left(1-o\left(1\right)\right)\E\left[\sum_{i\notin G}\ED\left(x^{i},y^{i}\right)\right].
\end{align*}
By subadditivity, the latter is at least the expected edit distance
between two long strings $x,y$ of respective lengths $\sum_{i\notin G}\left|u^{i}\right|,\sum_{i\notin G}\left|\hat{v}^{i}\right|$.
Assume without loss of generality that $\sum_{i\notin G}\left|u^{i}\right|\geq\sum_{i\notin G}\left|\hat{v}^{i}\right|$
(i.e. $x$ is longer). Let $y'$ be a uniformly random string of length
$\sum_{i\notin G}\left|u^{i}\right|$. We now have that
\begin{align}
\E\left[\sum_{i\notin G}\ED\left(x^{i},y^{i}\right)\right] & \geq\E\left[\ED\left(x,y\right)\right]\nonumber \\
 & \geq\E\left[\ED\left(x,y'\right)\right]-\left|\sum_{i\notin G}\left|u^{i}\right|-\sum_{i\notin G}\left|\hat{v}^{i}\right|\right|,\label{eq:ui<>vi}
\end{align}
where the second inequality follows by triangle inequality. 
\end{itemize}
Summing (\ref{eq:ui=00003Dvi}) and (\ref{eq:ui<>vi}), we have that
\[
\sum_{i=1}^{d_{\textsc{H}}}\ED\left(u^{i},\hat{v}^{i}\right)\geq\left(1+o\left(1\right)\right)\E\left[\ED\left(x,y'\right)\right],
\]
where $x,y'$ are uniformly random strings of length at least $\sum_{i\notin G}\left|u^{i}\right|=d'\cdot\Delta_{\textsc{Hamming}}\left(u_{\textsc{H}},v_{\textsc{H}}\right)$.
Finally, by subadditivity of the edit distance and Fekete's Lemma,
the expected distance between two random strings over their lengths
converges to some constant, so (for sufficiently large $d$), 
\[
\E\left[\ED\left(x,y'\right)\right]\geq\left(1-o\left(1\right)\right)\frac{\sum_{i\notin G}\left|u^{i}\right|}{d'}\lambda_{2,d'}^{\ED}.
\]
Hence also
\[
\sum_{i=1}^{d_{\textsc{H}}}\ED\left(u^{i},\hat{v}^{i}\right)\geq\left(1-o\left(1\right)\right)\lambda_{2,d'}^{\ED}\cdot\Delta_{\textsc{Hamming}}\left(u_{\textsc{H}},v_{\textsc{H}}\right).
\]

\subsubsection*{Derandomizing the construction}

In the construction above we used $2d=\Theta\left(\log N\log\log N\right)$
random bits to generate the strings $s_{c}^{i}$. We can reduce the
number of random bits to $\poly\log\log N$ by making the strings
$\log\log N$-wise independent (e.g. \cite{Kopparty13-notes}), over
which we can enumerate in $N^{o\left(1\right)}$ time. (Note that
there is another $O\left(\log\log N\right)$ factor because each string
requires $d'=O\left(\log\log N\right)$ random bits, and yet another
$\log d_{\textsc{H}}=O\left(\log\log N\right)$ random bits because
we generate $2d_{\textsc{H}}$ random strings.)

If the strings are pairwise independent, this already suffices to
guarantee that $\ED\left(s_{c}^{i},s_{c'}^{i'}\right)=\left(1\pm o\left(1\right)\right)\lambda_{2,d'}^{\ED}$
for every $s_{c}^{i},s_{c'}^{i'}$, with high probability. But when
we argued about the edit distance between any contiguous substring
of $s_{c_{1}}^{1}\circ\cdots\circ s_{c_{d_{\textsc{H}}}}^{d_{\textsc{H}}}$
and $s_{1-c_{i}}^{i}$, we assumed the strings $s_{c_{1}}^{1},\dots,s_{c_{d_{\textsc{H}}}}^{d_{\textsc{H}}}$
are all fully independent (in particular, we took a union bound over
all $2^{d_{\textsc{H}}}$ choices of $c_{1},\dots,c_{d_{\textsc{H}}}$).
Instead, we observe that it suffices to use the probabilistic argument
only for contiguous substrings of length $O\left(d'\right)$, i.e.
substrings that intersect with only a constant number of $s_{c_{i}}^{i}$:
for longer substring, their edit distance to $s_{1-c_{i}}^{i}$ is
always almost maximal, because the $s_{1-c_{i}}^{i}$ is much shorter
so we have to delete almost the entire longer substring. Therefore
using $\log\log N=\omega\left(1\right)$-wise independent strings
suffices.
\end{proof}

\subsubsection{Edit distance between random strings}
\begin{claim}
[Concetration of edit distance]\label{claim:concentration} Let $z\in\left\{ 0,1\right\} ^{n}$
be an arbitrary string, and let $x,y\in\left\{ 0,1\right\} ^{m}$
be chosen uniformly at random. Then for any $t\geq0$,
\[
\Pr_{x}\left[\left|\ED\left(z,x\right)-\E_{y}\left[\ED\left(z,y\right)\right]\right|>t\right]\leq2\exp\left(-\frac{2t^{2}}{m}\right).
\]
\end{claim}

\begin{proof}
For any choice of $x\in\left\{ 0,1\right\} ^{m}$ and $i\in\left[n\right]$,
let $\hat{x}^{i}\triangleq\left(x_{1},\dots x_{i-1},1-x_{i},x_{i+1},\dots,x_{m}\right)$.
Observe that by the triangle inequality, 
\[
\left|\ED\left(z,x\right)-\ED\left(z,\hat{x}^{i}\right)\right|\leq\ED\left(x,\hat{x}^{i}\right)=1.
\]
The claim follows from McDiarmid's inequality.
\end{proof}
\begin{claim}
[Expectation of edit distance]\label{claim:expectation} Let $z\in\left\{ 0,1\right\} ^{n}$
be an arbitrary string, and let $x\in\left\{ 0,1\right\} ^{n}$ be
chosen uniformly at random. Then 
\[
\frac{1}{n}\E\left[\ED\left(z,x\right)\right]\geq0.08686.
\]
\end{claim}

\begin{proof}
By \cite{Lueker09-LCS}, the expected length of the longest common
subsequence (Chvatal-Sankoff constant, $\gamma_{2}$) for sufficiently
large $n$ is at most 
\[
\frac{1}{n}\E\left[\LCS\left(z,x\right)\right]\leq0.826280.
\]
The claim follows from $\ED\left(z,x\right)\geq\left(n-\E\left[\LCS\left(z,x\right)\right]\right)/2$. 
\end{proof}

\subsection{Approximate Nearest Neighbor}

For completeness, we present the proof of Corollary \ref{cor:ANN}
from Theorems \ref{thm:main} and \ref{thm:edit}. Variants of this
reduction have appeared many times before; to the best of our knowledge,
the original idea is due to \cite{WW10-subcubic}.
\begin{proof}
[Proof of Corollary~\ref{cor:ANN}] We show that given an Approximate
Nearest Neighbor data structure with preprocessing time $O\left(N^{c}\right)$
and query time $O\left(N^{1-\delta}\right)$, we can obtain an algorithms
for {\sc Bichromatic Closest Pair} that runs in time $O\left(N^{2-\delta'}\right)$
for $\delta'=\delta'\left(\delta,c\right)$. Set $\gamma\triangleq1/2c$.
Given an instance $\left(A,B\right)$ of {\sc Bichromatic Closest Pair},
partition the set $A$ into $N^{1-\gamma}$ disjoint subsets $A_{1},\dots,A_{N^{\gamma}}$.
Now preprocess each set $A_{i}$ in time $O\left(\left(N^{\gamma}\right)^{c}\right)=O\left(N^{1/2}\right)$,
so $O\left(N^{3/2-\gamma}\right)$ in total for $O\left(N^{1-\gamma}\right)$
subsets. For each $b\in B$, query each of the data structures in
time $O\left(\left(N^{\gamma}\right)^{1-\delta}\right)=O\left(N^{\gamma-\gamma\delta}\right)$,
so $O\left(N^{2-\gamma\delta}\right)$ in total. Note that if there
exists a close pair $\left(a^{*},b^{*}\right)$, when we query the
subset that contains $a^{*}$ with vector $b^{*}$ we are guaranteed
to find a pair which is approximately as close.
\end{proof}
\bibliographystyle{alpha}
\bibliography{bib}

\end{document}